\documentclass[12pt,reqno]{amsart}
\usepackage{dsfont, amssymb,amsmath,amscd,latexsym, amsthm, amsxtra,amsfonts,enumerate}
\usepackage[all]{xy}
\usepackage[active]{srcltx}
\usepackage{mathrsfs}
\usepackage[all]{xy}
\usepackage[active]{srcltx}
\usepackage{graphicx}
\usepackage{bm}
\usepackage{bbm}
\usepackage{graphicx}
\usepackage{tabularx}
\usepackage{caption}
\usepackage{color}
\usepackage{algorithmicx,algorithm}
\usepackage{longtable}
\usepackage[round]{natbib}
\usepackage{ulem}
\usepackage{verbatim}
\usepackage{appendix}
\usepackage{subcaption}
\usepackage[utf8]{inputenc}
\usepackage[pdfstartview=FitH, bookmarksnumbered=true,bookmarksopen=true, colorlinks=true, pdfborder=001, citecolor=blue, linkcolor=blue,urlcolor=blue]{hyperref}
\usepackage{CJK}
\newtheorem{theorem}{Theorem}[section]
\newtheorem{lemma}[theorem]{Lemma}

\newtheorem{proposition}[theorem]{Proposition}

\newtheorem{Remark}{Remark}[section]

\newtheorem{assumption}{Assumption}[section]

\begin{document}
\makeatletter
\def\@setauthors{%
\begingroup
\def\thanks{\protect\thanks@warning}%
\trivlist \centering\footnotesize \@topsep30\p@\relax
\advance\@topsep by -\baselineskip
\item\relax
\author@andify\authors
\def\\{\protect\linebreak}%
{\authors}%
\ifx\@empty\contribs \else ,\penalty-3 \space \@setcontribs
\@closetoccontribs \fi
\endtrivlist
\endgroup } \makeatother
 \baselineskip 19pt
\title[{{\tiny Optimal information acquisition for eliminating estimation risk}}]
 {{\tiny Optimal information acquisition for eliminating estimation risk}}
 \vskip 10pt\noindent
\author[{Zongxia Liang, Qi Ye}]
{\tiny {\tiny  Zongxia Liang$^{a}$, Qi Ye$^{b}$ }
 \vskip 10pt\noindent
{\tiny Department of
Mathematical Sciences, Tsinghua University, Beijing 100084, China }
  \footnote{\\
 $ a$ email: liangzongxia@tsinghua.edu.cn\\
 $ b$ Corresponding author, email:   yeq19@mails.tsinghua.edu.cn\\  
 Funding: This work was funded by National Natural Science Foundation of China (Grant Nos.12271290)}}
\maketitle
\noindent

\begin{abstract}
This paper diverges from previous literature by considering the utility maximization problem in the context of investors having the freedom to actively acquire additional information to mitigate estimation risk. We derive closed-form value functions using CARA and CRRA utility functions and establish a criterion for valuing extra information through certainty equivalence, while also formulating its associated acquisition cost. By strategically employing variational methods, we explore the optimal acquisition of information, taking into account the trade-off between its value and cost. Our findings indicate that acquiring earlier information holds greater worth in eliminating estimation risk and achieving higher utility. Furthermore, we observe that investors with lower risk aversion are more inclined to pursue information acquisition. 
\end{abstract}

\section{Introduction}

In the financial market, investors are often assumed to possess complete information, allowing them to develop strategies that maximize their utility. However, in reality, obtaining complete information, especially regarding asset returns, is often uncertain, leading to estimation risk (\cite{kumar2008estimation}). \cite{gennotte1986optimal} and \cite{karatzas1991note} considers optimal consumption in an incomplete market setting, in \cite{lakner1995utility}, \cite{lakner1998optimal} and \cite{zohar2001generalized} the optimal terminal wealth is derived and the optimal strategy determined for linear Gaussian dynamics of the returns. \cite{karatzas2001bayesian} solves the problem if the return is a fixed random variable with known distribution. \cite{runggaldier2000stochastic} consider the optimization of the asymptotic growth rate for an infinite time horizon. \cite{sass2004optimizing} considers the return follows continuous time Markov chain.  To adapt to changing circumstances, investors continuously update their beliefs about asset returns over time as new asset prices arrive, employing Bayesian learning (\cite{liptser1977statistics}).

However, the estimation of asset returns is not solely updated based on asset prices. Investors have access to various additional sources of information, such as corporate earnings reports, macroeconomic indicators, political news or expert opinion (\cite{frey2012portfolio}). With the introduction of this extra information, it becomes intriguing to explore its impact on investment decisions and objective utility. While, most papers assume the information is given and usually it is of the same correlation of unknown part (\cite{xiong2010hetero}). Here, inspired by \cite{banerjee2020strategic}, we assume the investors can actively and dynamically to acquire the information in the market. While, this paper is originated by the classical paper \cite{kyle1985continuous} where the information is used to determine the ultimate payoffs's expectation and their object is to maximize the expected payoffs rather than the utility. But we still absorb the setting that we can acquire any information over time which can help to erase the estimation risk where the parameter is uncertain.

In this paper, we adopt the framework in \cite{karatzas2001bayesian} and equipped with extra information. We demonstrate that the correlated part of the extra information with the unobserved Brownian motion $W$ (the diffusion term of the risky asset's log price) can be considered as a small component of $W$. Importantly, this part is observed and can be utilized in return estimation. Consequently, the presence of extra information reduces estimation risk at a faster pace, leading to improved objective utility. Thus, this paper not only addresses the maximization of expected utility with any given extra information but also investigates the nature of the extra information itself, ultimately determining optimal information acquisition strategies.

We reveal that the extra information with higher correlation at any time can actually achieve higher utility. However, comparing utilities equipped with any extra information still requires solving the value function. To facilitate this analysis, we introduce the concept of an ``informative clock", which represents the ratio of the conditional variance of estimation to the risky asset's volatility. The informative clock quantifies the total volume of information and provides an intuitive index to assess the quality of the extra information. This perspective enhances the precision of estimation and allows for a concise presentation of the value function. It provides a more intuitive understanding of how extra information influences investment decisions and yields fruitful insights from various perspectives, such as attention, inattention, and information quality (\cite{huang2007rational} and \cite{veronesi2000does}).

We propose a criterion to quantify the value of extra information. Inspired by \cite{cabrales2013entropy} and \cite{kadan2019estimating}, we adopt the concept of equivalence of certainty. It means that taking endowment along with extra information can achieve the same utility as the case of taking endowment adding the value of information along with no extra information.  Our calculations demonstrate that the value of extra information is relatively independent of market conditions, including interest rates, volatility rates, and return estimates. Instead, it primarily depends on risk aversion and the informative clock. We show that the value of extra information remains constant across different initial wealth levels in the case of Constant Absolute Risk Aversion (CARA), providing a fixed value regardless of the endowment. Conversely, in the case of Constant Relative Risk Aversion (CRRA), investors focus on the multiple of wealth, where the value of information contributes a fixed proportional increase to the endowment. However, information acquisition and processing can incur significant costs in terms of time, effort, or expenses. To address this, we propose a penalty function for information acquisition based on the informative clock.

At last, we assume that any information characterized by an informative clock can be obtained in the market. Investors can strategically determine the optimal effort required to acquire such information for their investment decisions. We employ the variation method to solve the functional associated with the informative clock and derive its necessary condition. Our findings indicate that investors pay greater attention to information acquisition in the early stages rather than later stages and investors with less risk aversion tend to pay greater attention.

In summary, this paper makes three primary contributions. First, we establish a framework based on the concept of the "informative clock," enabling the concise solution of the value function. Second, we employ the concept of equivalence of certainty to determine the value of information. Lastly, we solve the optimal information acquisition strategy.

The remainder of this paper is structured as follows: Section 2 describes the market model and formulates the admissible space for investment strategies. Section 3 derives the value function and corresponding strategy under the assumption of given extra information, utilizing the concept of the informative clock,  presents the core analysis of the impact of extra information on investment decisions, evaluating its value and cost. Section 4 complements the content of information's source and informative clock's explosion. Lastly, the conclusion is drawn in the final section, with technical proofs and calculations provided in the appendix.

\vskip 15pt

\section{\bf Market Model and Problem Formulation}
 Suppose that  $(\Omega,\mathbb{F},\mathbb{P})$ is  a probability space equipped with the filtration $\mathbb{F}= \{\mathcal{F}_t: t\geq 0\} $ satisfying the usual conditions. We consider a financial market with one risk-free asset and one risky asset. The price of the risk-free asset $S_0=\{S_0(t), \ t\geq 0\}$ is given by
\begin{align*}
dS_0(t) = rS_0(t)dt, \ \ S_0(0)=1,
\end{align*}
and the price of the risky asset  $S=\{S(t), \ t\geq 0\}$
follows the stochastic differential equation (abbr.  SDE):
\begin{align*}
dS(t) = S(t)\left [\mu dt+\sigma dW_t\right ],\ \  S(0)=1,
\end{align*}
 where $W=\{ W_t: {t\geq 0}\}$ is a  Brownian motion  on $(\Omega,\mathbb{F},\mathbb{P})$,  $r$ and $\sigma$ are nonnegative constants. The drift term $\mu$ is an $\mathcal{F}_0$-measurable Gaussian random variable satisfying $\mu \sim \mathcal{N}(\mu_0,\sigma_0^2)$, and $\mu$ and $W$ are independent.
\vskip 5pt
The investor allocates the amounts $\pi_t$ of her wealth  into the risky asset at time $t$. Then the investor's total wealth as the self-financing process $X=\{ X_t,\ t\geq 0\}$ follows the following SDE:
\begin{align}
dX_t=rX_t+\pi_t(\mu-r)dt+\pi_t \sigma dW_t,\ \ t\geq 0, \ \  X_0=x_0.
\end{align}
\vskip 5pt
Traditionally, the investor formulates her investment strategy by observing the price of risky asset, i.e., the strategy $\pi=\{\pi_t:\ t\geq 0\}$ must be only adapted to the filtration $\mathbb{F}^S=\{\mathcal{F}^S_t: t\geq 0\}$, where $\mathcal{F}^S_t=\sigma \{ S_u: 0\leq u\leq t \}$, which is strictly smaller than $\mathcal{F}_t$, $ t\geq 0$.
\vskip 5pt
To capture the cross-sectional feature of unobserved process $W$, we introduce an extra process $m=\{m_t: t\geq 0\}$ in our model. It has to have some properties: it is observable, it is both Markov process and martingale from the background information which is independent to random variable $\mu$. The most valuable thing is that the process has correlation with the Brownian motion $W$:
\begin{align}\label{correlation}
\frac{d\langle m, W\rangle_t}{\sqrt{d\langle m \rangle_t d\langle W \rangle_t}}=\rho(t),
\end{align}
where $\rho(t)\in (-1,1)$ is determined function of time $t$. it is easy to see that $\rho(t)$ refers the correlation coefficient of two random variable $dW_t$ and $dm_t$ in the slight time from $t$ to $t+dt$. For  (\ref{correlation}) is well-defined, we  set  $\rho(t)=0$ when $\frac{d\langle m \rangle_t}{dt} = 0$.
\vskip 5pt
Now the investor can no longer only settle the strategy based on the observation of $S$, the process $m$ will be likely to provide more information for the more accurate estimation of $\mu$ where the ambiguity comes from. 
Mathematically speaking, we amplify the space of our investment strategy that $\pi$ can be adapted to the filtration $\mathbb{F}^S \vee \mathbb{F}^m$ rather than only adapted to the filtration $\mathbb{F}^S$. And we denote the space of admissible strategy as $\mathcal{A}_m:=\{ \pi | \int_0^T \pi_t^2 dt<\infty, \pi \ \mbox{is adapted to } \mathbb{F}^S \vee \mathbb{F}^m\}$, where $T>0$ refers the objective investment terminal time.
\vskip 5pt
The previous article's objective is to maximize the expectation utility of the terminal wealth at time $T$ over $\mathcal{A}_m$:
\begin{align}\label{solution2.3}
\sup_{\pi \in \mathcal{A}_m} EU(X_T),
\end{align} 
with the given extra information.

Our primary objective is to determine the optimal information to acquire in the market. To achieve this, we first need to establish what we mean by ``optimal''.  We must define a criterion to value the extra information, denoted as $Value(m)$, and formulate a cost function, denoted as $Cost(m)$, to quantify the expenses associated with acquiring this extra information. Consequently, our objective is to solve the optimal information acquisition strategy that maximizes the net value, defined as $Net(m) = Value(m) - Cost(m)$.

However, determining the criterion to value $Value(m)$ is not a straightforward task. One possible approach is to use equation (\ref{solution2.3}) to represent the value $Value(m)$. However, this approach has limitations. For instance, it may lead to negative values for $Value(m)$, which is counterintuitive. Additionally, comparing the value of information for investors with different risk aversions becomes challenging, as different utility functions cannot be directly compared. Despite these challenges, we still need to utilize the insights from equation (\ref{solution2.3}) to establish a criterion for valuing $Value(m)$, as a larger value in equation (\ref{solution2.3}) indicates greater value for the information. The specific criterion will be provided in a later section.

\section{\bf Extra information's value, cost and its optimal acquisition}

\subsection{The effect of correlation coefficient} The correlation coefficient of the extra information can be comprehended as a part of the unobserved $W$. A natural idea comes about that higher correlation coefficient means more valuable information of $W$. Then it can be conjectured that if two different sources of extra information have the same correlation coefficient with $W$, two corresponding optimal investment would achieve the same expected utility of terminal wealth. Similarly, if one extra information has the bigger correlation coefficient than the other one, then its corresponding expected utility would be bigger than the other. We now show the two conjectures
in the following Propositions \ref{Proposition31} and \ref{Proposition32}.

\begin{proposition}\label{Proposition31}
If there are two sources of extra information $m^1$ and $m^2$ sharing the same correlation coefficient that $\rho^1(t)=\rho^2(t)$ for all $t$, where $\rho^1(t):=\frac{d\langle m^1, W\rangle_t}{\sqrt{d\langle m^1 \rangle_t d\langle W \rangle_t}}$ and $\rho^2(t):=\frac{d\langle m^2, W\rangle_t}{\sqrt{d\langle m^2 \rangle_t d\langle W \rangle_t}}$. Denote $V_{m^1}:=\sup\limits_{\pi \in \mathcal{A}_{m^1}} EU(X_T)$ and $V_{m^2}:=\sup\limits_{\pi \in \mathcal{A}_{m^2}} EU(X_T)$,  similarly,  the superscript $1,2$ are indicating corresponding content in the case with information $m^1$ and $m^2$ (see (\ref{solution2.3})). Then $V_{m^1}=V_{m^2}$.
\end{proposition}
\vskip 5pt
\begin{proof}
	See Appendix A.
\end{proof}

Proposition \ref{Proposition31} tells us that not the exact form of the extra information but its correlation with $W$ indeed affect the objective function, thus it  provides  us some innate feeling that it is fine to ignore the exact form of extra information.

\vskip 5pt
\begin{Remark}
Here we can assume  $d\langle m \rangle_t=dt$. It is the reason that the class of the process $\left \{(H\cdot m)_t=\int_0^t H(s)dm_s \right.$, $H$ is a certain determined function with respect to time t and $H(t)\neq 0$ almost everywhere$\big\}$ shares the same information as $m$, i.e., $\mathbb{F}^m=\mathbb{F}^{(H\cdot m)}$. Moreover, this transformation makes sure (\ref{correlation}) still valid.  Then we can treat $m$ is a Brownian motion and (\ref{correlation}) is reduced into a simplified form as $d\langle m, W\rangle_t=\rho(t)dt$. In addition, we can assume $\rho(t)\geq 0$ if we replace new process $m'=\left\{m'_t:=\int_0^t \left[1_{\{\rho(t)\geq 0\}}\!\!-\!\!1_{\{\rho(t)< 0\}}\right ] dm_s\right\}$ with the process $m$ to assure the nonnegative and this setting is applied in the later content.	
\end{Remark}
\vskip 5pt

\begin{proposition}\label{Proposition32}
If there are two sources of extra information $m^1$ and $m^2$ sharing the correlation coefficient with the relation that $\rho^1(t) \geq \rho^2(t)$ for $ \forall$ $t$, then $V_{m^1} \geq V_{m^2}$.
\end{proposition}
\vskip 5pt
\begin{proof}
	See Appendix B.
\end{proof}

Proposition \ref{Proposition32} implies that higher correlation coefficient means the extra information is more valuable. Beyond that, we first put forward the idea about information dilution as the form $m^3$ being the mixture of the original information $m^1$ and a noise $W^{noise}$,  indeed, $W^{noise}$ has no use at all. Because $m^2$ contributes the same estimation of $\mu$ as $m^3$,  it contributes less than that $m^1$ contributes. So does it in the investment strategy.
\vskip 5pt
\begin{Remark}
The method proving Proposition \ref{Proposition32} corresponds the setting that $\rho(t)=0$ if $\frac{d \langle m \rangle_t}{dt}= 0$, which is explained as follows: When the equation holds, it means $dm_t=0$ at that instant. This implies that no valuable information is available at that instant. If we define the process $m'$:\!\! $m'_t=\int_0^t [1_{\{\frac{d \langle m \rangle_t}{dt}\neq 0\} }(s) dm_s+1_{\{\frac{d \langle m \rangle_t}{dt}= 0\} }(s) dW^{noise}_s ]$, as that both no information is equal to noise itself benefit nothing to strategy, then $m$ and $m'$ have the same effect to the process and investment.
\end{Remark}
\vskip 10pt

The extra information would bring about huge advantage for our strategy as what the former content reveal. But we still need to objectively measure the value of the extra information by a criterion rather than comparing two different extra information. Therefore, it is necessary to calculate Problem (\ref{solution2.3}) for preparation.

\subsection{\bf Optimal investment Problem (\ref{solution2.3}) with given extra information}
 The solution of control problem with ambiguity is usually based on bayesian learning method that we find the posterior probability distribution evolves over time by the data available up to each instant. Filtering method serves as an instrumental appliance to transform the estimation of parameter into the evolution of stochastic differential equation, which assists us in further analysis (\cite{liptser1977statistics}). Among them, the most widespread filtering is Kalman Bucy filtering, however, it can not be used directly. With the extra process $m$, the estimation of $\mu$ over time can not be given directly by two equations of conditional mean and variance as there must be aggregation of two sources of information. It is necessary to find the sufficiency statistic generated by $S$ and $m$ to fully estimate $\mu$.
\vskip 5pt
 In order to assure the deduction below carried out successfully, we need the following assumption on the extra observable process $m$.
\vskip 10pt
\begin{assumption}\label{assumption}
\begin{align}
\int_0^T \frac{1}{1-\rho^2(s)}ds<\infty.
\end{align}
\end{assumption}
\vskip 5pt
In the later subsection, we will illustrate the meaning of Assumption \ref{assumption} and discuss what happens in the absence of the assumption. We now present the following conditional expectation $E[\mu|\mathcal{F}_t^S\vee \mathcal{F}_t^m]$.
\vskip 5pt
\begin{theorem}\label{thm31}
\begin{align}
E[\mu|\mathcal{F}_t^S\vee \mathcal{F}_t^m]=\frac{1}{2}\sigma^2+\frac{y_0+\int_0^t q(t)^2 \left[dY_s-\sigma \rho(s)dm_s\right]}{t_0+\int_0^t q(s)^2 ds},
\end{align}
where $t_0:=\frac{\sigma^2}{\sigma_0^2}$, $y_0:=(\mu_0-\frac{1}{2}\sigma^2)t_0$ and $q(s):=\frac{1}{\sqrt{1-\rho(s)^2}}$, $Y$  is defined by (\ref{Y}) in Appendix C.
\end{theorem}

\begin{proof}
 See Appendix C.
\end{proof}
\vskip 5pt

 We find from the form of (\ref{posterior distribution}) in Appendix C that the conditional distribution of  $\mu $ given $ \mathcal{F}_t^S\vee \mathcal{F}_t^m$ is still a Gaussian, i.e.,
\begin{align}
	&\ \ \ \mu | \mathcal{F}_t^S\vee \mathcal{F}_t^m\\ \nonumber
	 &\sim \mathcal{N}\left (\frac{1}{2}\sigma^2+\frac{y_0+\int_0^t q(s)^2 [dY_s-\sigma \rho(s)dm_s]}{t_0+\int_0^t q(s)^2 ds},\right. 
	\left. \frac{\sigma^2}{t_0+\int_0^t q(s)^2 ds} \right).
\end{align}

To solve Problem (\ref{solution2.3}), based on Theorem \ref{thm31},  we introduce a new process $Z=\{Z_t,\  t\geq 0\}$ and the innovation process $\bar{W}=\{\bar{W}_t,\  t\geq 0\}$ as follows:
\begin{align*}
Z_t&:=\frac{1}{2}\sigma^2+\frac{y_0+\int_0^t q(s)^2 [dY_s-\sigma \rho(s)dm_s]}{t_0+\int_0^t q(s)^2 ds},\\
\bar{W}_t&:=\int_0^t \sigma^{-1}q(s) \left\{dY_s-\sigma \rho(s)dm_s-[E[\mu|\mathcal{F}_s^S\vee \mathcal{F}_s^m]-\frac{1}{2}\sigma^2]ds\right\}.
\end{align*}
Then $Z$ is adapted to the filtration $\{\mathcal{F}_t^S\vee \mathcal{F}_t^m: t\geq 0\}$ and  a sufficient statistic for estimating $\mu$ in this filtration, and $\bar{W}$ is a Brownian motion adapted to the filtration $\{\mathcal{F}_t^S\vee \mathcal{F}_t^m: t\geq 0\}$. Moreover, it is tested that $\bar{W}$  is independent to the Brownian motion $m$.

 Based on the definition of Brownian motion $\bar{W}$ and the explicit form of $E[\mu|\mathcal{F}_t^S\vee \mathcal{F}_t^m]$, we obtain that the  processes $X$, $Y$ and $Z$ satisfy the following SDEs:
\begin{align}
dX_t&=rX_t dt+\pi_t(Z_t-r) dt+\pi_t \sigma \rho(t)dm_t+\pi_t \sigma q(t)^{-1} d\bar{W}_t\\
dY_t&=(Z_t- \frac{1}{2}\sigma^2)dt+\sigma \rho(t)dm_t+\sigma q(t)^{-1} d\bar{W}_t,\\
dZ_t&=\frac{\sigma q(t) d\bar{W}_t}{t_0+\int_0^t q(s)^2 ds}.
\end{align}
\vskip 5pt
 Using the Markov property of the processes $X$ and $Z$, and the fact that $Z$ is a sufficient statistic of $\mu$, we establish  the following expected utility maximization problem  given extra information:
\begin{align}\label{hjbvalue}
V(t,x,z)=\sup_{\pi\in \mathcal{A}_m} E[U(X_T) | X_t=x, Z_t=z].
\end{align}
\vskip 10pt
It is easy to see that Problem (\ref {solution2.3}) is equivalent to Problem (\ref {hjbvalue}) with the initial state $(0,x_0,\mu_0)$.
\vskip 5pt
\begin{Remark}
 If we define $V(t,x,y,m):=\sup\limits_{\pi\in \mathcal{A}_m} E[U(X_T) | X_t=x, Y_t=y, m_t=m]$. The state value of $Y_t$ and $m_t$ can not give the estimation of $\mu$ as the entire process of $Y$ and $m$ has capability. Then the dynamic programming no longer hold for each instant the state's failure to estimate $\mu$. What's more that the value function $V(t,x,y,m)$ is not well-defined. That is why we must construct the process $Z$ to overcome this obstacle.
\end{Remark}
\vskip 5pt

To solve Problem (\ref {hjbvalue}), based on the definition of $V(t,x,z)$ and dynamic programming principle,
we start with analyzing the Hamilton-Jacobi-Bellman (HJB) equation satisfied by the value function $ V(t,x,z) $ as follows: \begin{align}\label{HJB}
\sup_\pi \mathcal{L}^\pi V(t,x,z)=0, \ \ \ V(T,x,z)=U(x),
\end{align}
where
\begin{align*}
\mathcal{L}^\pi V(t,x,z)&:=V_t+[rx+\pi(z-r)]V_x+\frac{1}{2}\sigma^2 \pi^2 V_{xx}\\
&+\sigma^2 \pi \frac{1}{t_0+\int_0^t q(s)^2 ds}V_{xz}\\
&+\frac{1}{2}q(t)^2\sigma^2\frac{1}{[t_0+\int_0^t q(s)^2 ds]^2} V_{zz}.
\end{align*}
We define a functional
\begin{align}\label{tau}
	\tau(\cdot):=t_0+\int_0^\cdot q(s)^2 ds.
\end{align}
 The explanation of the functional (\ref{tau}) will be given later. By using $\tau(\cdot)   $, the infinitesimal generator can be rewritten as the following  concise form:
\begin{align}\label{newinfini}
\mathcal{L}^\pi V(t,x,z)&:=V_t+[rx+\pi(z-r)]V_x+\frac{1}{2}\sigma^2 \pi^2 V_{xx}\nonumber\\
&+\sigma^2 \pi \frac{1}{\tau(t)}V_{xz}+\frac{1}{2}\sigma^2(-\frac{1}{\tau})'(t) V_{zz}.
\end{align}
Now we derive the  closed-form solutions of Problem (\ref {hjbvalue}) given extra information with CARA and CRRA utility functions, which is the foundation of the further discussion about the extra information and  the optimal information acquisition in Section 4.
\vskip 5pt
\begin{theorem}\label{valuefunction}
	(1) If $U(x)=-\frac{1}{\beta}e^{-\beta x}$ (CARA case), then we have
	\begin{align*}
		V\left(t,x,z\right)=U\left(e^{r(T-t)}x+\psi(t,z)\right),
	\end{align*}
	where
\begin{align*}
&\psi(t,z)=\frac{1}{2}a(t)(z-r)^2+c(t),\\	
&a(t)=\frac{1}{\beta \sigma^2}\frac{\tau(t)(T-t)}{\tau(t)+(T-t)},\\
&c(t)=\frac{1}{2\beta}\int_t^T\frac{\tau'(s)(T-s)}{\tau(s)[\tau(s)+(T-s)]}ds,
\end{align*}
	and the optimal investment strategy is
\begin{align*}
		\pi^*=e^{-r(T-t)}\frac{1}{\beta}\frac{\tau(t)}{\tau(t)+T-t}\frac{Z_t-r}{\sigma^2}.
\end{align*}
	(2) If $U(x)=\frac{x^{1-\gamma}}{1-\gamma}$ (CRRA case) where $\gamma$ satisfies the condition: $\gamma>1$ or $0<\gamma<1$ with $\frac{t_0}{T}>\frac{1-\gamma}{\gamma}$, then we have
	\begin{align*}
		V(t,x,z)=U(e^{r(T-t)}x)\exp\{\gamma \psi(t,z)\},
	\end{align*}
	where
\begin{align*}
&\psi(t,z)=\frac{1}{2}a(t)(z-r)^2+c(t),\\
&a(t)=\frac{1}{\gamma \sigma^2}\frac{\tau(t)\frac{1-\gamma}{\gamma}(T-t)}{\tau(t)-\frac{1-\gamma}{\gamma}(T-t)},\\
&c(t)=\frac{1}{2\gamma}\int_t^T  \frac{\tau'(s)\frac{1-\gamma}{\gamma}(T-s)}{\tau(s)[\tau(s)-\frac{1-\gamma}{\gamma}(T-s)]}ds,
\end{align*}
and the optimal investment strategy is
\begin{align*}
		\pi^*=\frac{\tau(t)}{\gamma \tau(t)-(1-\gamma)(T-t)}\frac{Z_t-r}{\sigma^2}X_t.
\end{align*}
	(3) If $U(x)=\frac{x^{1-\gamma}}{1-\gamma}$ (CRRA case) and $\gamma$ satisfies the condition $\frac{t_0}{T}\leq \frac{1-\gamma}{\gamma}$, then the value function will be infinite, i.e., the problem become meaningless. The occurrence of meaningless is not bring about by the extra information but bayesian problem itself.\\
(4) If $U(x)=\ln(x)$, the case with $\gamma$ approaching $1$, then
\begin{align*}
	V(t,x,z)=U(e^{r(T-t)}x)+\psi(t,z),
\end{align*}
   	where
\begin{align*}
	&\psi(t,z)=\frac{1}{2}a(t)(z-r)^2+c(t),\\	
	&a(t)=\frac{1}{\sigma^2}(T-t),\\	
	&c(t)=\frac{1}{2} \int_t^T  \frac{\tau'(s)(T-s)}{\tau^2(s)}ds,
\end{align*}
	and the optimal investment strategy is
	\begin{align*}
		\pi^*=\frac{Z_t-r}{\sigma^2}X_t.
	\end{align*}	
\end{theorem}
\begin{proof}
	See Appendix D.
\end{proof}
\vskip 5pt

Especially, if no extra information is involved, Problem (\ref {hjbvalue}) reduces to most viewed ambiguity case, which corresponds the situation the value function with $\tau(\cdot)=t_0+\cdot$; If no ambiguity is considered here, the value function will be given in the sense of limitation that $t_0= \infty$, $\tau(t)=\infty$ and $Z_t=\mu$. Then Problem (\ref {hjbvalue}) reduces to the classical problem.

Moreover, noticing the fact $\frac{\tau(t)}{\tau(t)+T-t}<1$, we can find with ambiguity, the investment strategy will be conservative than the classical case and less conservative if more information is involved in CARA case. Similarly, in CRRA case, it depends on the value of relative risk aversion $\gamma$. When $\gamma>1$, the investment strategy will be also conservative and less conservative if more information is involved. But when $\frac{T}{t_0+T}<\gamma<1$, the investment strategy will be aggressive and less aggressive if more information is involved. At last, when $\gamma=1$, the investment strategy remains same as the classical case.
\vskip 15pt

\subsection{Advantageous perspective: informative clock} Now we discuss the conditional variance  $var(t):=\frac{\sigma^2}{t_0+\int_0^t q(s)^2 ds}$  at time $t$, which is strictly decreasing over time, and
introduce the concept of informative clock as $\tau(t):=\frac{\sigma^2}{var(t)}=t_0+\int_0^t q(s)^2 ds$, an index measuring the volume of the information,
where $q(s):=\frac{1}{\sqrt{1-\rho(s)^2}}$ defined in Theorem \ref{thm31}.
\vskip 5pt

To comprehend this concept, we take an analogy. If we want to obtain a normal distribution's mean where the variance $\sigma^2$ is known. We get a sample as $X=(X_1,\cdots,X_N)$. Then the conditional distribution of $\mu$ given  $X$ is $\mathcal{N}(\bar{X},\frac{\sigma^2}{N})$, denoting this fact by $\mu | X \sim \mathcal{N}(\bar{X},\frac{\sigma^2}{N})  $. It can be seen that $N=\frac{\sigma^2}{var(\mu|X)}$, We make a similar comparison as the size of sample in this case and the informative clock in this article. It can be said that informative clock is actually an index measuring the volume of the information just like the size of the sample.
\vskip 5pt

The natural case is that $\rho(t)\equiv 0$ among where $m$ can be treated as noise. Then the estimate of $\mu$ over time is only estimated by the information of the price of risky asset itself. Correspondingly, $q(s)=1$ and $\tau(t)=t_0+t$. Informative clock keeps pace as the real time with a fixed distance as $t_0$. In fact, $t_0$ is the informative clock at real time $t=0$ which is actually the information volume hidden in the prior distribution. Imagining that there exists a duration of history, if the duration is of value $t_0$, the distribution of $\mu$ would be estimated with variance of $\sigma_0^2$ by information in this duration. Reversely thinking, that is why $t_0$ defined as $t_0:=\frac{\sigma^2}{\sigma_0^2}$ to match the duration of history. 
\vskip 5pt

Now the extra information intrudes into our system. It is natural to predict that it would bring about larger volume of information. This is obvious as $q(t)\geq 1$ and $\tau(t)=t_0+\int_0^t q(s)^2 ds$ where bigger correlation coefficient can bring higher increasing speed of the informative clock. The extra information serves as an acceleration of informative clock. It accelerates to reduce the variance of the estimation of $\mu$ than the real time, i.e.,  the extra information can eliminate estimation risk in fast pace. 
\vskip 5pt

The concept of informative clock is a more useful perspective for analyzing the value of information rather than correlation coefficient. If there are two sources of information, where the one has the information advantage than the other, i.e., with respect to the informative clock, the one is bigger than the other. The investor can make better decision to obtain higher utility. Moreover, Proposition \ref{proposition35} is a more strong conclusion than Proposition \ref{Proposition32} as it covers more cases.
\vskip 10pt

\begin{proposition}\label{proposition35}
If there are two sources of extra information $m^1$ and $m^2$ with informative clock relation that $\tau^1(t) \geq \tau^2(t)$ for $ \forall t\geq 0$, then  $V_{m^1} \geq V_{m^2}$.
\end{proposition}

\begin{proof}
	See Appendix E.
\end{proof}

 From the procedure of the comparison of $c$, we must broaden our mind. $\tau^1(t)\geq \tau^2(t)$ implies that in any real time the first information has more informative clock. However, if we take the perspective with $(\tau^1)^{-1}(u) \leq (\tau^2)^{-1}(u)$, it tells that in any same informative clock the first information has more remaining time to invest. Then we  see that the informative clock contributes the value function with two ways: the whole larger informative clock just as the first term in (\ref{compareofc}) shows and the adequate time to prepare in any same informative clock just as the second term shows.
\vskip 5pt

Another useful effect of informative clock is that HJB equation can be reduced into a representation with only informative clock as (\ref{newinfini}). In the infinitesimal generator, the term $V_t+[rx+\pi(z-r)]V_x+\frac{1}{2}\sigma^2 \pi^2 V_{xx}$ is the most classical HJB equation in the investment problem where $z$ is the certain rate of drift of the risky asset. Then $\sigma^2 \pi \frac{1}{\tau(t)}V_{xz}$ is comprehended as the adjustment of investment by the ambiguity of drift. The last term $\frac{1}{2}\sigma^2(-\frac{1}{\tau})'(t) V_{zz}$ depicts the evolution of the drift over time. Then HJB equation is intuitive under this perspective.
\vskip 5pt

Due to the close connection between $\rho$ and $\tau$ :  $\tau(t):=\frac{\sigma^2}{var(t)}=t_0+\int_0^t q(s)^2 ds$ and $q(s):=\frac{1}{\sqrt{1-\rho(s)^2}}$. For the convenience and vividness of expression, we use the functional $\tau$ to refer the information which has this volume of information. With it, all the value can be represented in concise form. And we assume the continuum of the $\rho(t)\in [0,1)$, i.e., any value in this interval can be obtained from the market. Thus we can assume any $\tau$ as the continuous function with derivative larger than one is obtained in the market. Finally, Proposition \ref{proposition35} has told us that the utility is irrelevant to the actual form of the source of information but what informative clock it can bring about for estimating $\mu$. Thus we can ignore the form of the extra information and only focus on the problem via the perspective of functional $\tau$.  \\
\vskip 10pt

\subsection{Value and cost of extra information} Now we solve the Problem (\ref {solution2.3}) and are equipped with the perspective of informative clock. To objectively measure the value of the extra information by a criterion, we introduce the concept of certainty equivalence of utility. As we have posed, we  define the value function $V_\tau(t,x,z)$ as the value function in the case that extra information can bring about the informative clock as the functional $\tau$. In the case without any extra information, the functional has the structure $\tau(\cdot)=t_0+\cdot$. Here we define it as the natural case and we denote that $\bold{0}$ as the functional in this case. It is the value function that we achieve by price itself. Then we define the value functional $Value(\cdot)$  by
\begin{align}\label{Val}
V_\tau(0,x_0,\mu_0)=V_{\bold{0}}\left(0,x_0+Value(\tau),\mu_0\right).
\end{align}
It means that taking initial value $x_0$ along with extra information with informative clock $\tau$, we can achieve the same utility as the case of dealing with initial value $x_0+Value(\tau)$ along with no extra information.  Here $V_{\bold{0}}$ is comprehended as our objective criterion.
\vskip 5pt

In the case of CARA utility function, comparing two value functions with the relationship $V_\tau(0,x_0,\mu_0) =V_{\bold{0}}\left(0,x_0+Value(\tau),\mu_0\right)$, we  obtain
\begin{align}
Value(\tau)=\int_0^T \frac{1}{2\beta} \frac{(T-t)\tau'(t)}{\tau(t)\left[\tau(t)+T-t\right]}dt-C_1,
\end{align}
where $C_1:=\int_0^T \frac{1}{2\beta} \frac{(T-t)}{(t+t_0)(T+t_0)}dt$ is a constant. This is actually the difference of $c(0)$ in two value functions with corresponding informative clock and it is more than zero just as Proposition \ref{proposition35} shows(see(\ref{compareofc})).
 We  see that it has nothing to do with the endowment $x_0$. This is derived from the property of CARA utility. People are caring about the absolute wealth in the utility function. Therefore the extra information will bring about the certain value.
 \vskip 5pt
In the case of CRRA utility function, comparing two value function with the relation  $V_\tau(0,x_0,\mu_0)=V_{\bold{0}}(0,x_0+Value(\tau),\mu_0)$, we  obtain
\begin{align}\label{CRRAVal}
& \ \ \ \ Value(\tau)\\ \nonumber
&=\left[\exp\left\{\int_0^T  \frac{1}{2(1-\gamma)} \frac{\tau'(t)\frac{1-\gamma}{\gamma}(T-t)}{\tau(t)[\tau(t)-\frac{1-\gamma}{\gamma}(T-t)]}dt\right\}/C_2-1\right]x_0.
\end{align}
where $C_2:=\exp\left\{\int_0^T  \frac{1}{2(1-\gamma)} \frac{\frac{1-\gamma}{\gamma}(T-t)}{(t+t_0)[(t+t_0)-\frac{1-\gamma}{\gamma}(T-t)]}dt\right\}$
 is a constant (The case $\gamma=1$ still holds in the sense of limit). Distinct from the CARA utility, people are caring about the multiple of the wealth with CRRA utility. The investors are surging to doubling their initial value. And we know from Eq.(\ref{CRRAVal}) and Proof of Proposition \ref{proposition35} that  the extra information will bring about the multiple increase of the endowment and the multiple number is the function of quotient of $c(0)$ with two informative clocks which coincides our experience.
  \vskip 5pt

No matter in CARA utility or CRRA utility, we find that the value of the extra information is independent with the market condition. It has nothing to do with the interest rate, risky asset's volatility and even the estimation of return. It only concerns about the informative clock which represents precision of the estimation and investor's preference for risk.

\begin{figure}
\begin{center}
\includegraphics[scale=0.6]{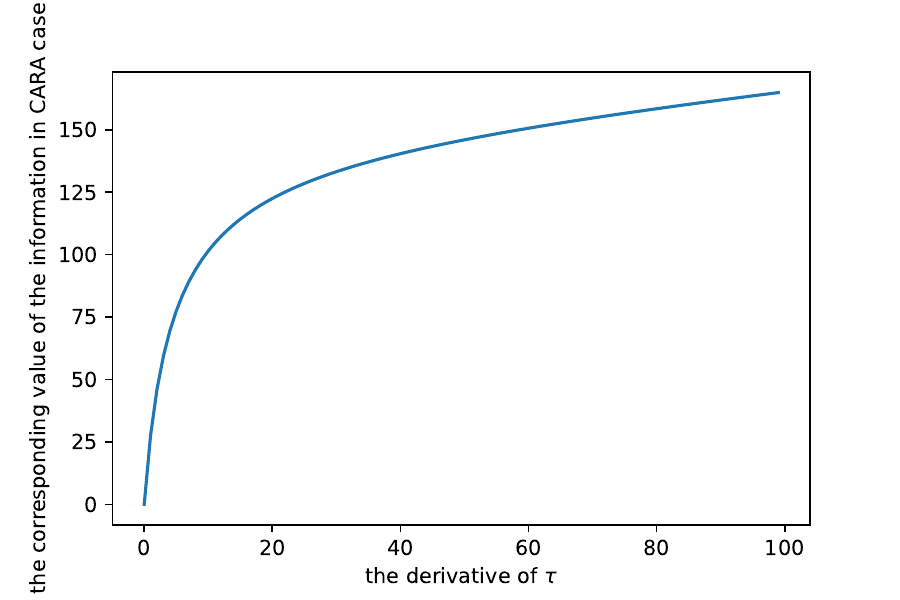}
\end{center}
\caption{The corresponding value of the information with respect to informative clock's derivative in CARA case.}
\end{figure}
\begin{figure}
\begin{center}
\includegraphics[scale=0.6]{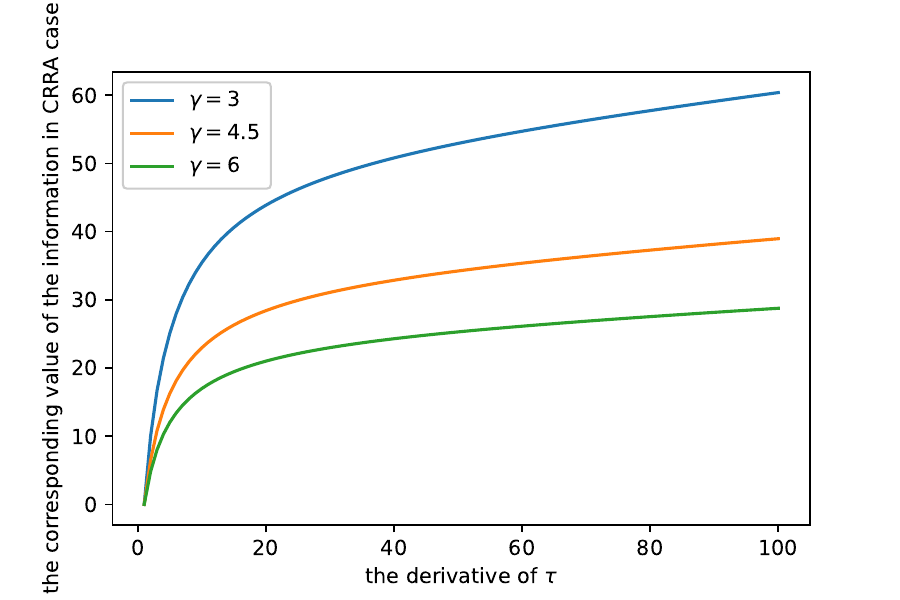}
\end{center}
\caption{The corresponding value of the information with respect to informative clock's derivative in CRRA case.}
\end{figure}
Here are two pictures of the value of information in CARA and CRRA case where the parameter is given with $t_0=4$, $T=2$, $\beta=0.001$, $x_0=1000$ and most importantly, $\tau(t)=t_0+k t$ where $k$ is a constant represent the derivative of the informative clock. For convenience, here we only use the linear informative clock for representation. It is obvious the information will bring huge value if the investor is less risk aversion both in CARA case and CRRA case where the value is inversely proportional to the absolute risk aversion in CARA case.
 \vskip 5pt

The other thing need to regard is that the marginal of volume of the information contributes less value when the information is plenty. Moreover, it has its bound which will be given in  Subsection \ref{4.6}.
 \vskip 5pt

Conversely, the gleaning of huge volume of information is costly. There must be a restriction of behavior that collecting information as much as the investor is not allowed. We define the punishment functional $Cost(\cdot)$ as follows:
\begin{align}\label{COS}
Cost(\tau)=\int_0^T cost\left(\tau'(t)\right)dt,
\end{align}
where the function $cost: [1,\infty)\to \mathrm{R}^+$ is an increasing and convex function. As we know, $\tau'(t)=\frac{1}{\sqrt{1-\rho(t)^2}} \geq 1$ and when $\tau'(t)=1$, it means $\rho(t)=0$ and no useful extra information is involved. Thus, we  set the boundary condition: $cost(1)=0$. The monotone property  and convexity originate from the fact that the effort to glean the marginal information ($\tau'(t)$) is positive and increases w.r.t.the volume of information.

\subsection{The most worthwhile information to acquire}

The value and cost of information are well defined, see (\ref{Val}) and (\ref{COS}). In this subsection the most important thing  we consider is to balance the value and cost to determine the most worthwhile information to acquire, i.e., solving the following functional maximization problem:
\begin{align}\label{FMP}
&\sup_{\tau} Net(\tau),\\
&Net(\tau):= Value(\tau)-Cost(\tau),\nonumber
\end{align}
where the informative clock $\tau$ is differentiable  and  its derivative is larger than one everywhere.
 \vskip 5pt
Now we assume  
$cost(x)=\lambda (x-1)^2$ for an instance, where $\lambda>0$ is the  relative ability of acquiring information,
 and solve Problem (\ref{FMP}) as follows.
 \vskip 5pt
\begin{theorem}\label{MTHM}
 In the case of CARA utility function, the optimal informative clock satisfies the condition
\begin{align}\label{carael}
\tau''(t)+\frac{1}{4\beta \lambda} \frac{1}{[\tau(t)+T-t]^2}=0.
\end{align}
In the case of CRRA utility function, the optimal informative clock satisfies the condition
\begin{align}\label{crrael}
\tau''(t)+\frac{y}{4\gamma \lambda} \frac{1}{[\tau(t)-\frac{1-\gamma}{\gamma}(T-t)]^2}=0,
\end{align}
where $y$ is a positive number.
\end{theorem}
\vskip 5pt

\begin{proof}
	See Appendix F.
\end{proof}

\vskip 5pt
\begin{Remark}
In fact, using  the  same method as in Theorem \ref{MTHM}  we can solve Problem (\ref{FMP}) for  other punishment function. Theorem \ref{MTHM} just provides a template to find extreme point of Problem (\ref{FMP}). And this template is the most widespread for describing the punishment function.
\end{Remark}
\vskip 5pt

The last work to find the optimal informative clock is to search all the functional satisfying the necessary condition with the different initial condition $\tau'(0)$ and $y$. This part is trivial and we omit it here.
\vskip 5pt
 Based on Eq.(\ref{carael}) and Eq.(\ref{crrael}), we have
  $$\tau''(t)+g(t)=0,$$
  where $g(t)>0$ (it holds for any punishment functional).  It follows that
\begin{align*}
	\tau''(t)<0.
\end{align*}

\begin{figure}
\begin{center}
\includegraphics[scale=0.6]{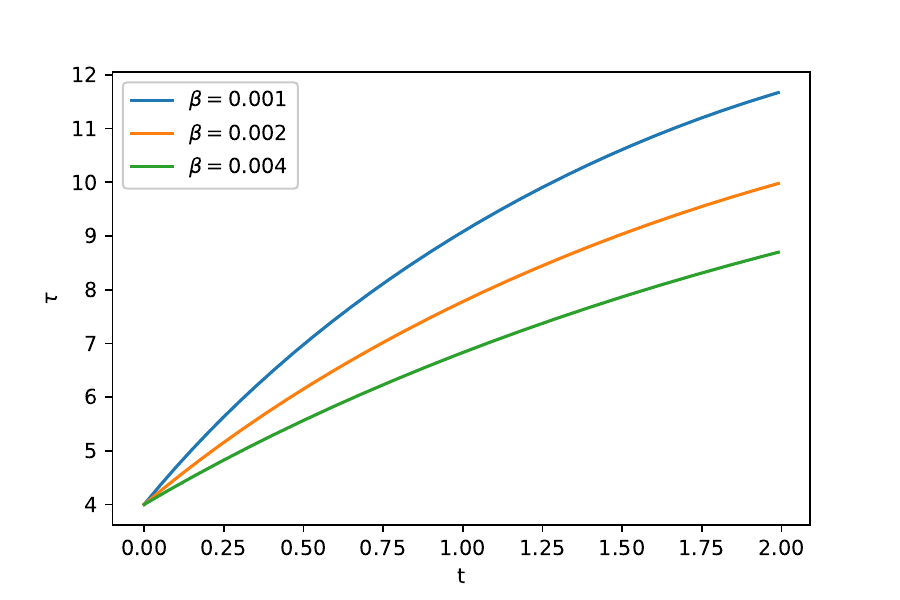}
\end{center}
\caption{The optimal informative clock in CARA case.}
\end{figure}
\begin{figure}
\begin{center}
\includegraphics[scale=0.6]{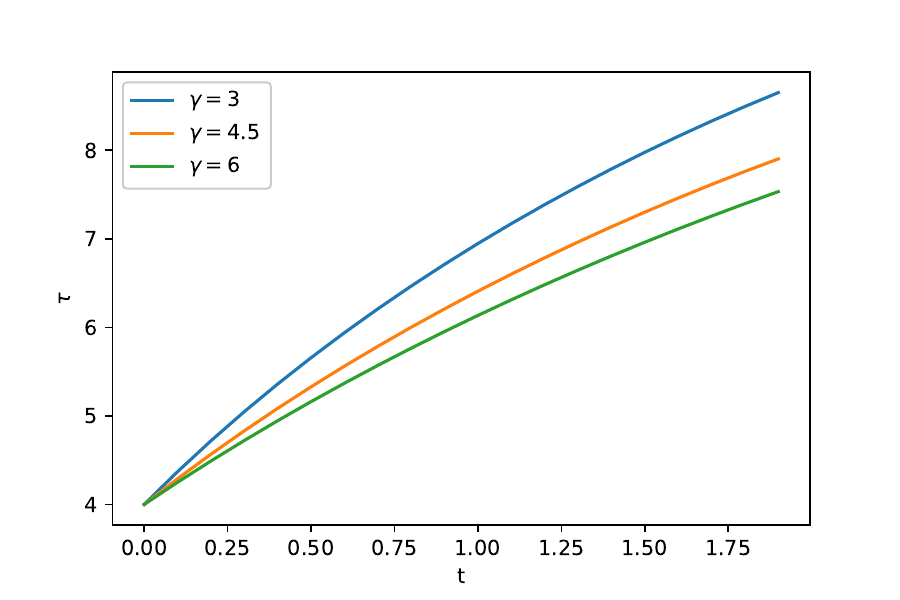}
\end{center}
\caption{The optimal informative clock in CRRA case.}
\end{figure}

That is to say the investor will devote less effort to glean the information as the time goes. In fact, this conclusion can also be derived from Proposition \ref{proposition35}. If there is a information with $\tau'(t_1)<\tau'(t_2)$ where $t_1<t_2$ having the most net value,  we reshuffle the order of $\tau'$ as $\hat{\tau}'$ and define a new $\hat{\tau}(t)=\int_0^t \hat{\tau}'(s)ds$ to achieve $\hat{\tau}(t) \geq \tau(t)$ for all $t$. Then $V^{\hat{\tau}}> V^{\tau}$ and $Value(\hat{\tau})> Value(\tau)$. However, their cost is the same. Then we find a new informative clock with more net value and there is a contradiction. Thus we must glean the information as soon as possible.
\vskip 5pt
Here are two pictures of the optimal informative clock in CARA and CRRA case where the parameter is given the same as $t_0 = 4$, $T = 2$, $x_0 = 1000$, $\lambda = 1$. As we can see, the optimal informative is really increasing and concave. Furthermore, if we compare the optimal informative clock for the investor with different risk aversion in CARA case or CRRA case, we will see investor with lower risk aversion are more inclined to glean more information. In other words, gleaning information benefits more to the investor with lower risk aversion.

\section{Complements of information itself}
\subsection{Sources of information}

We provide a discussion around the extra information itself. In the financial market, a wide range of information can be gathered and utilized. One commonly considered type of information is systemic risk, which we assume is generated by the process $W^{\text{systemic}}$. For instance, we can consider $m$ to be equal to $W^{\text{systemic}}$. However, it is important to note that any observed risk can be represented by the process $m$. Whenever a risky asset exists in the market, various risk factors contribute to the price fluctuations and can be observed. Furthermore, the phenomenon of co-movement in asset prices offers us a perspective to identify the process $m$.

One potential concern might arise regarding the calculation of $\rho(t)$ as $W$ itself is not directly observable. However, since we have the observation $Y$ (defined in (\ref{Y})), and $\rho(t)=\frac{d\langle m, Y\rangle_t}{\sqrt{d\langle m \rangle_t d\langle Y \rangle_t}}$, we can effectively calculate their correlation coefficient without direct observation of $W$. This equation provides a feasible way to compute the correlation coefficient, which empowers us to test different sources of information. Therefore, we have the ability to assess the correlation between $Y$ and different sources of information, allowing us to identify the ones with the highest correlation. Furthermore, we can even combine multiple sources of information through signal engineering techniques to achieve a higher correlation with $Y$.
\vskip 15pt

\subsection{the explosion of informative clock} \label{4.6} Assumption \ref{assumption} (see (\ref{assumption})) can be rewritten as $\tau(T)<\infty$. It is the case that the informative clock is finite from $0$ to $T$. It arouses our curiosity that what would happen when the informative clock become infinite. In fact, at that time, the ambiguity hidden in $\mu$ is totally erased. As the conditional variance of $\mu$ becomes zero and the distribution of $\mu$ has collapsed into a single point's distribution. While this phenomenon has been supposed as  there is a `insider' investor, who can observe both the drift and driving Brownian motion $W$ (\cite{zhao1999bayesian}) . The informative clock becomes infinite means that the investor becomes `insider' investor knowing everything.
\vskip 5pt

It is natural to put forward some problems around that the informative clock evolve from finite time to the infinite. The main one is that how it happens. In fact, if we extend the range of $\rho$ with $\rho(t)\in [0,1]$. It will happen when the fact $\rho(t)$ approaches or equals to one (for example, $m=W$) with enough time which contributes infinite value for the informative clock. Another one may be around the function of informative clock becomes singular as it tends infinite in finite time. There is no worry as in the HJB equation, we can observe that $\tau$ only exists with the form $\frac{1}{\tau}$ which can be continuous from $\frac{1}{\tau(0)}$ to $0$ where two ends are both finite. Thus, it really bothers nothing when there is a surge of infinite information if we assure the continuity of $\frac{1}{\tau}$. Moreover, all the deduction in  Section 3 still hold if we use limitation as there is always informative clock exist in the dominator to assure anything well-posedness. Above all, everything remains alright if we just take $(\frac{1}{\tau})(t)=0$ when $\tau(t)=0$.

 \vskip 5pt
Ultimately, we take $\tau(0^+)=\infty$ for consideration (can be comprehended as the limit $\tau(t)=t_0+nt$ where $n\to \infty$), 
that is, the investor is assumed as `insider' investor at the beginning (time is $0^+$). This corresponds the value of $V(0^+,x_0,Z_{0^+})$ which is the classical case with $\mu=Z_{0^+}$. And $\mu$ has the prior distribution $\mathcal{N}(\mu_0,\sigma^2_0)$. Then we have that $\int_{\mathrm{R}} V(0^+,x_0,\mu_0+u) \mathcal{N}_{\sigma_0^2}(u)du$ is actually the objective value what we want. In fact, this value can be gotten directly in the value function if we allow the generalized function calculation in $c(0)$ in the HJB equation. Taking CARA for example,
\begin{align*}
	c(0)&=\frac{1}{2\beta}\int_0^T\frac{\tau'(s)(T-s)}{\tau(s)[\tau(s)+(T-s)]}ds\\
	&=\frac{1}{2\beta}\int_0^{0^+}\left[\frac{\tau'(s)}{\tau(s)}-\frac{[\tau(s)+T]'}{\tau(s)+T}\right]ds\\
	&=\frac{1}{2\beta}\left[\ln(\tau(s))-\ln(\tau(s)+T)\right]|_0^{0^+}\\
	&=\frac{1}{2\beta}\ln(\frac{t_0+T}{t_0}).
\end{align*}
Similarly, we  obtain the value
\begin{align*}
	c(0)=\frac{1}{2\gamma}\ln(\frac{t_0}{t_0-\frac{1-\gamma}{\gamma}T})
\end{align*}
in CRRA case if $\gamma\neq 1$. And
\begin{align*}
	c(0)=\frac{1}{2}\frac{T}{t_0}
\end{align*}
in CRRA case when $\gamma=1$.
\vskip 5pt
 As regards to Proposition \ref{proposition35} and content of information's value, it is obvious that the value of any extra information has its bound. And the bound is finite and can be approached if we pay enough cost.
\\
\vskip 10pt
\section{\bf Conclusions}
In this paper, we show that,  with the extra information, the estimation of the risky asset's is preciser and the utility can be improved in addition. We first put forward the new concept of ``informative clock'' to describe the precision of the estimation and use it to represent everything in the concise form. The value and cost are well defined for any extra information. From the process of trading off, it tells us to glean the extra information earlier will be great and we find investors with less risk aversion are more inclined to glean information. Finally, it bothers us nothing to permit the informative clock infinite and find the bound of the extra information's value.
\vskip 10pt

\appendix
\vskip 10pt

\section*{Appendix A.  Proof of Proposition \ref{Proposition31}}
\begin{proof}
Without loss of generality, $m^1$ and $m^2$ are regarded as Brownian motion (see in the Remark).  As two processes $m^1$ and $m^2$ sharing the same correlation coefficient with $Y$, by using Ito formula, the characteristic functions of the two pairs $(Y, m^1)$ and $(Y, m^2)$  are equal, as such, they have the same probability law. In the case containing $m^1$, as the optimal $\pi^1$ is adapted to the filtration $\mathcal{F}_t^S\vee \mathcal{F}_t^{m^1}$,  there must exist the functional $f$ s.t. $\pi^1_t=f(Y_u,m^1_u,0\leq u\leq t)$. We choose $\pi^2_t=f(Y_u,m^2_u,0\leq u\leq t)$ adapted to the filtration $\mathcal{F}_t^S\vee \mathcal{F}_t^{m^2}$. Then the process classes $(Y, m^1, \pi^1, X^1)$ and $(Y, m^2, \pi^2, X^2)$ have the same probability law. As such,  $X_T^1$ and $X_T^2$ have also the same probability law, and their expected utilities are equal. As $\pi^1$ is the optimal one and $\pi^2$ is an selected one,   $V_{m^1} \geq V_{m^2}$. Similarly, we have $V_{m^2} \leq V_{m^1}$. Thus, $V_{m^1}= V_{m^2}$.
\end{proof}

\section*{Appendix B.  Proof of Proposition \ref{Proposition32}}
Before proving the second conjecture, we prepare a lemma  as follows:
\vskip 5pt
\begin{lemma}\label{Lemma4.2}
Let $(\Omega,\mathcal{F},\mathcal{P})$ be the probability space with $\mathcal{G}_1$ and $\mathcal{G}_2$ being the sub-$\sigma$-algebras of $\mathcal{F}$. For any random variable $X\in L^1( \Omega,\mathcal{F},\mathcal{P}    )$, if $X$ and $\mathcal{G}_1$ are independent of $\mathcal{G}_2$, then we have
\begin{align}\label{expectationreduce}
E[X|\mathcal{G}_1\vee \mathcal{G}_2]=E[X|\mathcal{G}_1].
\end{align}
\end{lemma}
\vskip 5pt
\begin{proof}
As both sides of (\ref{expectationreduce})are $\mathcal{G}_1\vee \mathcal{G}_2$-measurable, and the nonempty collection $\mathfrak{F}_0\triangleq\{ A_1\cap A_2 | A_1\in \mathcal{G}_1,A_2\in \mathcal{G}_2\}$ of subsets of $\Omega$ is a $\pi$-system generating $\mathcal{G}_1\vee \mathcal{G}_2$, Eq.(\ref{expectationreduce}) is equivalent to
\begin{align*}
\int_{A_1\cap A_2}  E[X|\mathcal{G}_1\vee \mathcal{G}_2]dP=\int_{A_1\cap A_2} E[X|\mathcal{G}_1]dP
\end{align*}
for $\forall$ $A_1\cap A_2\in \mathfrak{F}_0 $ . As $E[X1_{A_1}|\mathcal{G}_1]$ and $1_{ A_2}$ are independent, and $X1_{A_1}$ and $1_{ A_2}$ are independent, we have
\begin{align*}
& \ \  \ RHS\\
&=\int_{\Omega} 1_{\{A_1\cap A_2\}} E[X|\mathcal{G}_1]dP=\int_{\Omega}1_{A_1}1_{ A_2} E[X|\mathcal{G}_1]dP\\
&=\int_{\Omega} E[X1_{A_1}|\mathcal{G}_1] 1_{ A_2}dP=E[ E[X1_{A_1}|\mathcal{G}_1]] E[1_{ A_2}]\\
&=E[ X1_{A_1}] E[1_{ A_2}]=\int_{\Omega} X1_{A_1}1_{ A_2} dP\\
&=\int_{\Omega}1_{A_1\cap A_2}  E[X|\mathcal{G}_1\vee \mathcal{G}_2]dP=\int_{A_1\cap A_2}  E[X|\mathcal{G}_1\vee \mathcal{G}_2]dP.
\end{align*}
\end{proof}
\vskip 5pt
\begin{proof}
 We can construct artificially a pure noise Brownian motion $W^{noise}$ which is assumed to be on the probability space $(\Omega,\mathbb{F},\mathbb{P})$ and adapted to the filtration $\mathbb{F}$. If not, there is an extension $(\widetilde{\Omega},\widetilde{\mathbb{F}},\widetilde{\mathbb{P}})$ of $(\Omega,\mathbb{F},\mathbb{P})$ which can be defined by the cartesian product of the original one and the probability space of $W^{noise}$. Additionally, the extension of space makes no difference of investing and objective function as the restriction of $\pi$. Thus the probability space can be arbitrarily fertile for our need. Similarly, $m^1$ and $m^2$ are still regarded as Brownian motions.

 Next we claim that the observed noise contributes nothing to our strategy for any given extra information, that is,
\begin{align}\label{utilitynoise}
\sup\limits_{\pi \in \mathcal{A}_{m}} EU(X_T)=\sup\limits_{\pi \in \mathcal{A}_{m,W^{noise}}} EU(X_T),
\end{align}
where $\mathcal{A}_{m,W^{noise}}:=\{ \pi | \int_0^T \pi_t^2 dt<\infty, \pi \text{ adapted to } \mathbb{F}^S \vee \mathbb{F}^m\vee \mathbb{F}^{W^{noise}}\}$.

Let us review the estimation of $\mu$. In fact,  if we let $X=\mu$, $\mathcal{G}_1= \mathcal{F}_t^S\vee \mathcal{F}_t^m$ and $\mathcal{G}_2= \mathcal{F}_t^{noise}$ in Eq.(\ref{expectationreduce}), we have
\begin{align}\label{munoise}
E[\mu|\mathcal{F}_t^S\vee \mathcal{F}_t^m\vee \mathcal{F}_t^{W^{noise}}]=E[\mu|\mathcal{F}_t^S\vee \mathcal{F}_t^m].
\end{align}
Eq.(\ref{munoise}) tells us that the noise give totally no information of $\mu$. Using the same method, we can  construct innovation process $\bar{W}$ and sufficient process $Z$, in which no differences exist with or without $W^{noise}$. Thus, (\ref{utilitynoise}) holds. Taking $m=m^1$, we obtain
\begin{align*}
\sup\limits_{\pi \in \mathcal{A}_{m^1}} EU(X_T)=\sup\limits_{\pi \in \mathcal{A}_{m^1,W^{noise}}} EU(X_T).
\end{align*}
Define the process $m^3:=\{m^3_t, t\geq 0\}$:
\begin{align}
m^3_t=\int_0^t \left\{\frac{\rho^2(s)}{\rho^1(s)}dm^1_s+\sqrt{1-\left[\frac{\rho^2(s)}{\rho^1(s)}\right]^2} dW^{noise}_s\right\}.
\end{align}
where $\frac{\rho^2(s)}{\rho^1(s)}:=0$ if $\rho^1(s)=0$, $\rho^2(s)=0$. Then $m^3$ is a Brownian motion with the property that $\rho^3(t):=d\langle m,W \rangle_t=\rho^2(t)$. Using Proposition \ref{Proposition41}, we have
\begin{align*}
\sup\limits_{\pi \in \mathcal{A}_{m^2}} EU(X_T)=\sup\limits_{\pi \in \mathcal{A}_{m^3}} EU(X_T).
\end{align*}
As $\mathcal{A}_{m^3} \subset  \mathcal{A}_{m^1,W^{noise}}$, based on the definition of supreme, we obtain
\begin{align*}
\sup\limits_{\pi \in \mathcal{A}_{m^3}} EU(X_T)\leq \sup\limits_{\pi \in \mathcal{A}_{m^1,W^{noise}}} EU(X_T).
\end{align*}
Thus
\begin{align*}
\sup\limits_{\pi \in \mathcal{A}_{m^1}} EU(X_T) \geq \sup\limits_{\pi \in \mathcal{A}_{m^2}} EU(X_T),
\end{align*}
that is,  $V_{m^1}\geq V_{m^2}$.
\end{proof}
\vskip 5pt

\section*{Appendix C.  Proof of Theorem \ref{thm31}}
\begin{proof}
\vskip 5pt
\begin{align}\label{Y}
dY_t=(\mu-\frac{1}{2}\sigma^2)dt+\sigma dW_t,\ t\geq 0, \ \  Y_0=0.
\end{align}
Define  $n=\{n_t:=\int_0^t q(s)[dW_s-\rho(s)dm_s],\ 0\leq t\leq T\}$.  Assumption \ref{assumption} guarantees that the process $n$ is well-defined.
 Furthermore,  $n$ is  a Brownian motion and orthogonal to the process $m$. And we have the following  orthogonal decomposition of $W$:
 \begin{align*}
dW_t=\rho(t) dm_t+q(t)^{-1} dn_t.
\end{align*}
Using SDE(\ref{Y}), we obtain
\begin{align*}
\left[dY_t-\sigma \rho(t)dm_t+\frac{1}{2}\sigma^2dt\right]=\mu dt+\sigma q(t)^{-1} dn_t.
\end{align*}
Define the posterior distribution:
\begin{align*}
p(u,t)du:=P(\mu \in du | \mathcal{F}_t^S\vee \mathcal{F}_t^m),\ \forall u\in \mathrm{R}.
\end{align*}
Using the same detailed calculation  as in \cite{wonham1964some} , we have
\begin{align}\label{evolution}
&\ \ \  p(u,t)du\\ \nonumber
&= p(u,0)du\exp\left\{-\frac{1}{2}u^2\int_0^t \sigma^{-2}q(s)^2 ds\right.\\ \nonumber
&\left. +u \int_0^t \sigma^{-2}q(t)^2 \left[dY_s-\sigma \rho(s)dm_s+\frac{1}{2}\sigma^2ds\right]\right\} \\ \nonumber
&\bigg/ \int_{\mathrm{R}} p(u,0)du\exp\left \{-\frac{1}{2}u^2\int_0^t \sigma^{-2}q(s)^2 ds\right.\\ \nonumber
&\left.+u \int_0^t \sigma^{-2}q(t)^2 \left[dY_s-\sigma \rho(s)dm_s+\frac{1}{2}\sigma^2ds\right]\right \}.\nonumber
\end{align}
The prior distribution of $\mu$ yields
\begin{align*}
p(u,0)du=\frac{1}{\sqrt{2\pi \sigma_0^2}}e^{-\frac{(u-\mu_0)^2}{2 \sigma_0^2}}.
\end{align*}
Putting it in (\ref{evolution}), we obtain
\begin{align}\label{posterior distribution}
&\hskip 0.3cm p(u,t)du\\ \nonumber
&=\frac{1}{\sqrt{2\pi \frac{\sigma^2}{t_0+\int_0^t q(s)^2 ds}}} \times \\ \nonumber
&\exp\left\{-\frac{\left[u-\left(\frac{1}{2}\sigma^2+
\frac{y_0+\int_0^t q(t)^2 (dY_s-\sigma \rho(s)dm_s)}{t_0+\int_0^t q(s)^2 ds}\right)\right]^2}{2 \frac{\sigma^2}{t_0+\int_0^t q(s)^2 ds}}\right\}.\nonumber
\end{align}
Therefore
\begin{align}
& \ \ \ E[\mu|\mathcal{F}_t^S\vee \mathcal{F}_t^m]=\int_{\mathrm{R}} u p(u,t)du\\ \nonumber
&=\frac{1}{2}\sigma^2+\frac{y_0+\int_0^t q(t)^2 \left[dY_s-\sigma \rho(s)dm_s \right]}{t_0+\int_0^t q(s)^2 ds}.
\end{align}
\end{proof}
\vskip 15pt

\section*{Appendix D.  Proof of Theorem \ref{valuefunction}}
\begin{proof}
\vskip 5pt
(1) CARA case: utility function $U(x)=-\frac{1}{\beta}e^{-\beta x}$. We use the ansatz:
\begin{align*}
V(t,x,z)=U(e^{r(T-t)}x+\psi(t,z)).
\end{align*}
Then the HJB equation  corresponding to $ V(t,x,z)$ becomes the following  PDE of $\psi$:
\begin{align}\label{HJBA2}
0= &-\beta V \left[ \psi_t+\pi^*(z-r)e^{r(T-t)}-\beta \frac{1}{2}\sigma^2 \pi^{*2}e^{2r(T-t)}\right.\\
&\left.  -\beta \sigma^2 \pi^* \frac{1}{\tau(t)}e^{r(T-t)}\psi_z -\beta \frac{1}{2}\sigma^2 (-\frac{1}{\tau})'(t)\psi_z^2+\frac{1}{2}\sigma^2 (-\frac{1}{\tau})'(t)\psi_{zz}\right],\nonumber\
\end{align}
where
\begin{align*}
\pi^*=e^{-r(T-t)}\frac{(z-r)-\beta \sigma^2 \frac{1}{\tau(t)}\psi_z}{\beta \sigma^2}.
\end{align*}
Putting $\pi^*$ into (\ref{HJBA2})  yields
\begin{align}\label{psi}
\psi_t+\frac{\beta\left[ \frac{1}{\beta}(z-r)-\sigma^2 \frac{1}{\tau(t)}\psi_z\right]^2}{2\sigma^2}+
\frac{1}{2}\sigma^2 (-\frac{1}{\tau})'(t)[-\beta \psi^2_z+\psi_{zz}] =0.
\end{align}
We guess that $ \psi  $ has the following form:
\begin{align*}
\psi(t,z)=\frac{1}{2}a(t)(z-r)^2+b(t)(z-r)+c(t).
\end{align*}
Substituting the form of $ \psi$ into Eq.(\ref{psi}), we have
\begin{align}
&\frac{1}{2}a'(t)+\frac{\beta \left[ \frac{1}{\beta}-\sigma^2 \frac{1}{\tau(t)}a(t) \right]^2}{2\sigma^2}\\ \nonumber
&+\frac{1}{2}\sigma^2 (-\frac{1}{\tau})'(t)(-\beta)a(t)^2=0, \ \ a(T)=0,\\ 
&b'(t)+(\cdots)b(t)= 0,\ \ \ \ b(T)=0, \\
&c'(t)+(\cdots)b(t)+\frac{1}{2}\sigma^2 (-\frac{1}{\tau})'(t)(a(t))=0,\ \  c(T)=0.
\end{align}
It follows that $\psi(T,z)=0$ and $b(t)=0$ for $t\in [0,T]$. To get $a$,
 we define
\begin{align*}
  m(t)= \frac{1}{\beta}-\sigma^2 \frac{1}{\tau(t)}a(t).
\end{align*}
Using the ODE of $a(t)$,
\begin{align*}
\tau(t)m'(t)+\beta\left[\tau'(t)-1\right]m(t)^2-\tau'(t)m(t)=0,
\end{align*}
which is a Riccati equation. Let $d(t)=\frac{1}{m(t)}$, the Riccati equation is transformed into the following ODE:
\begin{align*}
d'(t)=-\frac{\tau'(t)}{\tau(t)}d(t)+\frac{\beta[\tau'(t)-1]}{\tau(t)}.
\end{align*}
Thus
\begin{align*}
d(t)&=\beta+\beta(T-t)\frac{1}{\tau(t)}, \\  
m(t)&=\frac{1}{\beta+\beta(T-t)\frac{1}{\tau(t)}},\\
a(t)&=\frac{\tau(t)}{\sigma^2}\left[\frac{1}{\beta}-\frac{1}{\beta+\beta(T-t)\frac{1}{\tau(t)}}\right]\\
&=\frac{1}{\beta \sigma^2}\frac{\tau(t)(T-t)}{\tau(t)+(T-t)},\\
c(t)&=\int_t^T \frac{1}{2}\sigma^2 (-\frac{1}{\tau})'(s)(a(s))ds\\
&=\frac{1}{2\beta}\int_t^T\frac{\tau'(s)(T-s)}{\tau(s)[\tau(s)+(T-s)]}ds.
\end{align*}

(2) As CRRA case is intricacy, we show the proof in three different categories:  The first one is $U(x)=\frac{x^{1-\gamma}}{1-\gamma}$ with the necessary condition that $\frac{t_0}{T}>\frac{1-\gamma}{\gamma}$, the second one with $\frac{t_0}{T} \leq \frac{1-\gamma}{\gamma}$ and the last one is  $U(x)=\ln(x)$.

(i) If $U(x)=\frac{x^{1-\gamma}}{1-\gamma}$ with $\frac{t_0}{T}>\frac{1-\gamma}{\gamma}$,  we  use the ansatz:
\begin{align*}
V(t,x,z)=U(e^{r(T-t)}x)\exp\{\gamma \psi(t,z)\}.
\end{align*}
The HJB equation corresponding to $V(t,x,z)  $ is equivalent to the following:
\begin{align*}\nonumber
0= \gamma V &  \left[ \psi_t+\frac{1-\gamma}{\gamma}\frac{\pi^*}{x}
(z-r)-\frac{1}{2}\sigma^2 (\frac{\pi^*}{x})^2 (1-\gamma)\right. \nonumber\\
&\left.+\sigma^2 \frac{\pi^*}{x}\frac{1-\gamma}{\tau(t)}\psi_z+\frac{1}{2}\sigma^2 (-\frac{1}{\tau})'(t)\left[\gamma\psi^2_z+\psi_{zz}\right]  \right],
\end{align*}
where
\begin{align*}
\pi^*=x\frac{\sigma^2 \frac{1}{\tau(t)}\psi_z+\frac{1}{\gamma}(z-r)}{\sigma^2}.
\end{align*}
Then
\begin{align*}
\psi_t+\frac{(1-\gamma)\left[ \sigma^2 \frac{1}{\tau(t)}\psi_z + \frac{1}{\gamma}(z-r)\right]^2}{2\sigma^2}\\
+\frac{1}{2}\sigma^2 (-\frac{1}{\tau})'(t)\left[\gamma \psi^2_z+\psi_{zz}\right] =0.
\end{align*}
We guess
\begin{align*}
\psi(t,z)=\frac{1}{2}a(t)(z-r)^2+b(t)(z-r)+c(t).
\end{align*}
Similar to that of CARA case,
\begin{align*}
&\frac{1}{2}a'(t)+\frac{(1-\gamma)\left[ \sigma^2 \frac{1}{\tau(t)}a(t)+\frac{1}{\gamma}\right]^2}{2\sigma^2}\\
&+\frac{1}{2}\sigma^2 (-\frac{1}{\tau})'(t)\gamma a(t)^2=0,\  a(T)=0,\\
&b(t)\equiv 0,\ b(T)=0, \\
&c'(t) +\frac{1}{2}\sigma^2 (-\frac{1}{\tau})'(t)(a(t))=0, \ c(T)=0.
\end{align*}
As for the equation of $a$,  first  making a  substitution:
\begin{align*}
m(t)= \sigma^2 \frac{1}{\tau(t)}a(t)+\frac{1}{\gamma},
\end{align*}
  we get a Riccati equation of $m$ as follows:
\begin{align}\label{RRE}
\tau(t)m'(t)+\left[ (1-\gamma)+\gamma\tau'(t)\right] m(t)^2-\tau'(t)m(t)=0.
\end{align}
Using a substitution: $d(t)=\frac{1}{m(t)}$,    the Riccati equation is the following solvable form:
\begin{align*}
d'(t)=-\frac{\tau'(t)}{\tau(t)}d(t)+\frac{(1-\gamma)+\gamma\tau'(t)}{\tau(t)}.
\end{align*}
Solving the last ODE,
\begin{align*}
d(t)=\gamma-(1-\gamma)(T-t)\frac{1}{\tau(t)},\\
 m(t)=\frac{1}{\gamma-(1-\gamma)(T-t)\frac{1}{\tau(t)}}.
\end{align*}
Then
\begin{align*}
a(t)&=\frac{\tau(t)}{\sigma^2}\left[ \frac{1}{\gamma-(1-\gamma)(T-t)\frac{1}{\tau(t)}}-\frac{1}{\gamma} \right]\\
&=\frac{1}{\gamma \sigma^2}\frac{\tau(t)\frac{1-\gamma}{\gamma}(T-t)}{\tau(t)-\frac{1-\gamma}{\gamma}(T-t)},\\
c(t)&=\int_t^T \frac{1}{2}\sigma^2 (-\frac{1}{\tau})'(s)(a(s))ds\\
&=\frac{1}{2\gamma}\int_t^T  \frac{\tau'(s)\frac{1-\gamma}{\gamma}(T-s)}{\tau(s)[\tau(s)-\frac{1-\gamma}{\gamma}(T-s)]}ds.
\end{align*}
Thus we obtain that closed-form of the value function is solved in the first situation. Moreover, if $\gamma>1$, the value function will always be bounded. But if $0<\gamma<1$, the necessary condition that $\frac{t_0}{T}>\frac{1-\gamma}{\gamma}$ assures the boundedness of the value function.\\
\vskip 5pt
(ii) If $U(x)=\frac{x^{1-\gamma}}{1-\gamma}$  with $\frac{t_0}{T} \leq \frac{1-\gamma}{\gamma}$, the value function will be infinite, the above method of ``ansatz" does not work. Here we take the strategy that $\pi_t=k X_t$. We will show that if $k\to \infty$, the utility expectation will tend to infinite.

Indeed, if $\pi_t=k X_t$,  the terminal state  is
\begin{align*}
X_T=e^{rT}x_0\exp\left\{\left[k(\mu-r)-\frac{1}{2}\sigma^2k^2\right]T+\sigma k W_T\right\}.
\end{align*}
As $\mu$ and $W_T$ are independent, we have
\begin{align*}
&E\left\{U(X^k_T)/U(e^{rT}x_0)\right\}\\
=&E\left\{\exp\left[(1-\gamma)(k(\mu-r)-\frac{1}{2}\sigma^2k^2)T+(1-\gamma) \sigma k W_T\right]\right\}\\
=&\exp\left\{(1-\gamma)k(\mu_0-r)T-\frac{1}{2}(1-\gamma)\sigma^2k^2T \right.\\
&\ \ \  \ \ \ +\left. \frac{1}{2}(1-\gamma)^2 \sigma^2 k^2 T+ \frac{1}{2}(1-\gamma)^2 k^2 \sigma^2 T^2 \frac{1}{t_0} \right\}\\
=&\exp\left\{(1-\gamma)k(\mu_0-r)T+\frac{1}{2}(1-\gamma)^2 \sigma^2 k^2 T\left[\frac{T}{t_0}-\frac{\gamma}{1-\gamma}\right] \right\}.
\end{align*}
 It follows that  $EU(X^k_T)\to \infty$ as  $k\to \infty$. Thus, the original problem is meaningless no matter the extra information is included or not. The phenomenon happens when the investor has less relative risk aversion and the ambiguity hidden in drift term is really too much that the extreme situation where the drift term is far beyond the normal life weighted a lot in probability which cause the expectation utility infinite.
Moreover,  if $\frac{T}{t_0}-\frac{\gamma}{1-\gamma}< 0$, then
\begin{align*}
	&\ \ \ \sup_k E\{U(X^k_T)\} \\
	&=U(e^{rT}x_0)\exp\left\{ \frac{1}{2}\frac{T}{\frac{\gamma}{1-\gamma}-\frac{T}{t_0}}\frac{(\mu_0-r)^2}{\sigma^2} \right\}\\
	&=U(e^{rT}x_0) \exp\left\{\gamma \cdot \frac{1}{2}a(0)(\mu_0-r)^2\right\}.
\end{align*}
From which we see that $a(t)$ is a important function to weight the informative clock and remaining time to the maturity. That is why with any extra informative $a(0)$ is the same which will benefit the discussion of value of information later. In this perspective, the calculation of $a(t)$ in fact is not only by trick.
\vskip 5pt
(iii) If $U(x)=\ln(x)$,  we use the ansatz:
\begin{align*}
V(t,x,z)=U(e^{r(T-t)}x)+\psi(t,z).
\end{align*}
Similar to that of (i), we have
\begin{align*}
&\psi_t+(z-r)(\frac{\pi^*}{x})-\frac{1}{2}\sigma^2 (\frac{\pi^*}{x})^2 +\frac{1}{2}(-\frac{1}{\tau})'\psi_{zz}=0,
\\
&\pi^*=x\frac{z-r}{\sigma^2}.
\end{align*}
And
\begin{align}
&\psi(t,z)=\frac{1}{2}a(t)(z-r)^2+c(t),\nonumber\\
&a(t)=\frac{1}{\sigma^2}(T-t),\nonumber\\
&c(t)=\int_t^T \frac{1}{2}\sigma^2 (-\frac{1}{\tau})'(s)(a(s))ds=\frac{1}{2} \int_t^T  \frac{\tau'(s)(T-s)}{\tau^2(s)}ds.
\end{align}
\end{proof}

\section*{Appendix E.  Proof of Proposition \ref{proposition35}}
\begin{proof}
If we want to compare the value of $V_{m^1}$ and $V_{m^2}$, we only need to compare the value of $V^1(0,x_0,\mu_0)$ and $V^2(0,x_0,\mu_0)$.
\vskip 5pt
In fact, $a(0)$ is the value determined by $\tau(0)$ and $T$, which is irrelevant to any extra information in both CARA case and CRRA case (see Proof of Theorem \ref{valuefunction}(3)).
The only comparison of value function will focus on the value of $c(0)$. We make the comparison of $c(0)$ with different informative clock by some transformations.
\vskip 5pt
In the case of CARA for example, based on Theorem \ref{valuefunction}, we have
\begin{align*}
	c(0)&=\frac{1}{2\beta}\int_0^T\frac{\tau'(s)(T-s)}{\tau(s)\left[\tau(s)+(T-s)\right]}ds\\
	&=\frac{1}{2\beta}\int_0^T\left[\frac{1}{\tau(s)}-\frac{1}{\tau(s)+T-s}\right]\tau'(s)ds\\
	&=\frac{1}{2\beta}\int_0^{\tau(T)} \left[\frac{1}{u}-\frac{1}{u+T-\tau^{-1}(u)}\right]du,
\end{align*}
where the last equation is by the variable transformation $u=\tau(s)$. Then
\begin{align}\label{compareofc}
	& \ \ \ \ \ c^1(0)-c^2(0)\\
	&=\frac{1}{2\beta}\int_{\tau^2(T)}^{\tau^1(T)} \left[\frac{1}{u}-\frac{1}{u+T-(\tau^1)^{-1}(u)}\right]du\\
	&\ \ \ \  +\frac{1}{2\beta}\int_0^{\tau^2(T)} \left[\frac{1}{u+T-(\tau^2)^{-1}(u)}-\frac{1}{u+T-(\tau^1)^{-1}(u)}\right]du\nonumber\\
& \geq 0+0=0, \nonumber
\end{align}
where $c^1(0)$ and $c^2(0)$ are in the value function $V^1$ and $V^2$ respectively.
The second term is deduced from the fact $(\tau^1)^{-1}(u) \leq (\tau^2)^{-1}(u)$, where $\tau^{-1}$ is the reverse function of $\tau$. Then
\begin{align*}
	& \ \ \ \ V^1\left(0,x_0,\mu_0\right)\\
	&=U\left(e^{rT}x_0+\frac{1}{2}a(0)(\mu_0-r)^2+c^1(0)\right)\\
	&\geq U\left(e^{rT}x_0+\frac{1}{2}a(0)(\mu_0-r)^2+c^2(0)\right)\\
	&=V^2\left(0,x_0,\mu_0\right).
\end{align*}
This is what we desire.
\vskip 5pt
Similarly, the conclusion can still be proven by the same way as in CRRA case. When $\gamma\neq 1$.
\begin{align*}
	c(0)&=\frac{1}{2\gamma}\int_0^T  \frac{\tau'(s)\frac{1-\gamma}{\gamma}(T-s)}{\tau(s)\left[\tau(s)-\frac{1-\gamma}{\gamma}(T-s)\right]}ds\\
	     &=\frac{1}{2\gamma}\int_0^T\left[\frac{1}{\tau(s)-\frac{1-\gamma}{\gamma}(T-s)}-\frac{1}{\tau(s)}\right]\tau'(s)ds\\
	     &=\frac{1}{2\gamma}\int_0^{\tau(T)}\left[\frac{1}{u-\frac{1-\gamma}{\gamma}[T-\tau^{-1}(u)]}-\frac{1}{u}\right]du.
\end{align*}
\vskip 5pt
Similarly, If $\gamma<1$, we have $c^1(0)\geq c^2(0)$.  But if $\gamma>1$, we have $c^1(0)\leq c^2(0)$. But it doesn't matter for the comparison of value function if we notice that the utility functions $U(x)$ are positive or negative with  $\gamma <1$ or $\gamma >1$ . Hence, we obtain
\begin{align*}
	& \ \ \ \ V^1\left(0,x_0,\mu_0\right)\\
	&=U(e^{rT}x_0)\exp\left\{\gamma \cdot \left[\frac{1}{2}a(0)(\mu_0-r)^2+c^1(0)\right]\right\}\\
	&\geq U(e^{rT}x_0)\exp\left\{\gamma \cdot \left[\frac{1}{2}a(0)(\mu_0-r)^2+c^2(0)\right]\right\} \\
	&=V^2\left(0,x_0,\mu_0\right).
\end{align*}
At last, when $\gamma=1$,
\begin{align*}
	c(0)&=\frac{1}{2} \int_0^T  \frac{\tau'(s)(T-s)}{\tau^2(s)}ds\\
	&=\frac{1}{2} \int_0^{\tau(T)}  \frac{T-\tau^{-1}(u)}{u^2}du
\end{align*}
It is  easy to see  the relation with $c^1(0)\geq c^2(0)$,  we also have
\begin{align*}
		& \ \ \ \ V^1\left(0,x_0,\mu_0\right)\\
		&=rT +\ln(x_0)+\frac{1}{2}a(0)(\mu_0-r)^2+c^1(0)\\
	&\geq rT +\ln(x_0)+\frac{1}{2}a(0)(\mu_0-r)^2+c^2(0)\\
	&=V^2\left(0,x_0,\mu_0\right).
\end{align*}
Thus, we end the proof.
\end{proof}

\section*{Appendix F. Proof of Theorem \ref{MTHM}}
\begin{proof}
In the case of CARA utility function, if we define
\begin{align*}
\mathfrak{L}(t,\tau,\tau'):=\frac{1}{2\beta} \frac{(T-t)\tau'(t)}{\tau(t)[\tau(t)+T-t]}- \lambda [\tau'(t)-1]^2,
\end{align*}
then  Problem (\ref{FMP}) is equivalent to solving $\sup\limits_{\tau}\int_0^T \mathfrak{L}(t,\tau(t),\tau'(t))dt$. In fact, based on the method of variations, we  obtain that the optimal functional has to satisfy the necessary condition as the Euler-Lagrange equation:
\begin{align*}
\frac{d}{dt}\frac{\partial}{\partial \tau'}\mathfrak{L}(t,\tau(t),\tau'(t))=\frac{\partial}{\partial \tau}\mathfrak{L}(t,\tau(t),\tau'(t)).
\end{align*}
It is simplified as
\begin{align}
\tau''(t)+\frac{1}{4\beta \lambda} \frac{1}{[\tau(t)+T-t]^2}=0.
\end{align}\\

\vskip 5pt

In the case of CRRA utility function with $\gamma\neq 1$, it will become a little complex because the Euler-Lagrange method will become invalid due to the nonlinear part of the integration. We use the method of  duality to make the integration linear.
 \vskip 5pt
Noticing that $\frac{x_0}{C_2}\exp(x)$ has the duality function $y-y\ln(\frac{C_2 y}{x_0})$, we have
\begin{align}\label{CRRAVall}
\frac{x_0}{C_2}\exp(x)=\sup_{y>0} \left\{y-y\ln(\frac{C_2 y}{x_0})+xy\right\}.
\end{align}
By using Eq.(\ref{CRRAVal}) and Eq.(\ref{CRRAVall}), the objective function $Net(\tau)$ is
\begin{align*}
Net(\tau)&=\sup_{y>0} \left\{y-y\ln(\frac{C_2 y}{x_0})-x_0\right.\\
&\left.+y \left[\int_0^T  \frac{1}{2(1-\gamma)} \frac{\tau'(t)\frac{1-\gamma}{\gamma}(T-t)}{\tau(t)[\tau(t)-\frac{1-\gamma}{\gamma}(T-t)]}dt\right] \right.\\
&\left. -\int_0^T \lambda [\tau'(t)-1]^2 dt\right\}.
\end{align*}
If we define the
\begin{align*}
& \ \  \mathfrak{L}^y(t,\tau,\tau')\\
&:=\frac{y}{2(1-\gamma)} \frac{\tau'(t)\frac{1-\gamma}{\gamma}(T-t)}{\tau(t)[\tau(t)-\frac{1-\gamma}{\gamma}(T-t)]}- \lambda [\tau'(t)-1]^2,
\end{align*}
then
\begin{align*}
Net(\tau)=\sup_{y>0} \left\{y-y\ln(\frac{C_2 y}{x_0})-x_0+ \int_0^T \mathfrak{L}^y(t,\tau,\tau') dt\right\}.
\end{align*}
As such,
\begin{align*}
& \ \ \ \ \sup_{\tau} Net(\tau)\\
&=\sup_{\tau} \sup_{y>0} \left\{y-y\ln(\frac{C_2 y}{x_0})-x_0+ \int_0^T \mathfrak{L}^y(t,\tau,\tau') dt\right\}\\
&= \sup_{y>0} \sup_{\tau}\left\{y-y\ln(\frac{C_2 y}{x_0})-x_0+ \int_0^T \mathfrak{L}^y(t,\tau,\tau') dt\right\}.
\end{align*}
Then we  fix $y>0$ and solve
\begin{align*}
\sup_{\tau}\left\{y-y\ln(\frac{C_2 y}{x_0})-x_0+ \int_0^T \mathfrak{L}^y(t,\tau,\tau') dt\right\}
\end{align*}
by Euler-Lagrange equation with
\begin{align*}
\frac{d}{dt}\frac{\partial}{\partial \tau'}\mathfrak{L}^y(t,\tau(t),\tau'(t))=\frac{\partial}{\partial \tau}\mathfrak{L}^y(t,\tau(t),\tau'(t)).
\end{align*}
It follows that
\begin{align}
\tau''(t)+\frac{y}{4\gamma \lambda} \frac{1}{[\tau(t)-\frac{1-\gamma}{\gamma}(T-t)]^2}=0.
\end{align}	
Moreover, for $U(x)=\ln(x)$, as the calculation is similar and the outcome coincide (\ref{crrael}) if we treat $\gamma=1$,   we do not need to carefully distinguish the case whether $\gamma$ equals to one.
\end{proof}

\end{document}